\documentclass[aps,pra,twocolumn,superscriptaddress,nofootinbib]{revtex4-2}
\pdfoutput=1
\usepackage[utf8]{inputenc}
\usepackage[english]{babel}
\usepackage[T1]{fontenc}
\usepackage{amsmath}

\usepackage{amsthm,amssymb,amscd, physics}
\usepackage{color}
\usepackage{soul}
\usepackage{url}
\usepackage[normalem]{ulem}
\usepackage{enumitem}
\usepackage[left=2.5cm,right=2.5cm,bottom=2.5cm,top=2.5cm]{geometry}
\usepackage{tikz}
\usepackage{csquotes}
\usepackage{float}
\usepackage{fancyhdr}
\usepackage{mathtools}
\usepackage{stmaryrd}
\usepackage{thm-restate}
\usepackage[hidelinks]{hyperref}
\usepackage{diagbox}
\usetikzlibrary{patterns}

\newcommand{\cket}[1]{|#1\rangle}
\newcommand{\bracket}[2]{\langle #1|#2\rangle}
\newcommand{\mcl}[1]{\ensuremath{\mathcal{#1}}}
\newcommand{\norme}[2]{\left|\left|#1\right|\right|_{#2}}

\newcommand{\Acal}{{\mcl A}}
\newcommand{\Bcal}{{\mcl B}}

\newcommand{\EcalKDC}{{\mcl E}_{\mathrm{KD+}}}
\newcommand{\EcalKDCpu}{{\mcl E}_{\mathrm{KD+}}^{\mathrm{pure}}}

\newcommand{\spanRAB}{{\mathrm{span}_{\R}(\Acal\cup\Bcal)}}
\newcommand{\spanR}{{\mathrm{span}_{\R}}}
\newcommand{\Hcal}{{\mcl H}}

\newcommand{\Ncal}{{\mcl N}}

\newcommand{\Ker}{{\mathrm{Ker}}}

\newcommand{\Ucal}{{\mcl U}}

\newcommand{\N}{\mathbb{N}}
\newcommand{\Z}{\mathbb{Z}}
\newcommand{\R}{\mathbb{R}}

\newcommand{\C}{\ensuremath{\mathbb{C}}}

\newcommand{\bbone}{\mathbb{I}}

\newcommand{\mab}{m_{\Acal,\Bcal}}

\newcommand{\conv}[1]{\mathrm{conv}\left(#1\right)}

\newcommand{\VR}{V_{\mathrm{KDr}}}

\renewcommand{\Re}[1]{\mathrm{Re}#1}
\renewcommand{\Im}[1]{\mathrm{Im}#1}
\newcommand{\Ran}[1]{\mathrm{Ran}\left(#1\right)}
\renewcommand{\H}{\mcl{H}}
\newcommand{\SAH}{\mathcal S(\mathcal H)}
\newcommand{\SAO}[1][d]{\mcl{S}_{#1}(\C)}

\newcommand{\MatC}[1]{\ensuremath{\mcl{M}_{#1}(\C)}}
\newcommand{\MatR}[1]{\mcl{M}_{#1}(\R)}
\newcommand{\IntEnt}[2]{\llbracket #1 , #2 \rrbracket}
\newcommand{\convAB}{\conv{\Acal \cup \Bcal}}
\renewcommand{\epsilon}{\varepsilon}

\newtheorem{Theorem}{Theorem}

\newtheorem{Lemma}[Theorem]{Lemma}
\newtheorem{Prop}[Theorem]{Proposition}

\newcommand{\Uex}{\ensuremath{{U_{\star}}}}

\newcommand{\Triv}{\ensuremath{\mathcal{T}}}
\newcommand{\NT}{\ensuremath{\mathcal{N}}}
\newcommand{\rcu}[1][U]{\ensuremath{\widetilde{C}^{#1}}}
\newcommand{\id}[1][d]{\ensuremath{I_{#1}}}

\newcommand{\inpr}[2]{\ensuremath{\langle #1, #2 \rangle}}
\definecolor{prettyblue}{RGB}{047,103,177}
\definecolor{prettyred}{RGB}{191,044,035}

\begin{document}

\title{Almost no experiments have classical Kirkwood-Dirac representations}

\author{Christopher Langrenez}
\email{christopher.langrenez@univ-lille.fr}
\affiliation{Univ. Lille, CNRS, Inria, UMR 8524, Laboratoire Paul Painlev\'e, F-59000 Lille, France}
\author{Wilfred Salmon}
\affiliation{Hitachi Cambridge Lab., J.J. Thomson Avenue, Cambridge CB3 0US, UK}
\affiliation{DAMTP, Centre for Mathematical Sciences, Univ. of Cambridge, CB30WA, UK}
\author{Stephan De Bi\`evre}

\affiliation{Univ. Lille, CNRS, Inria, UMR 8524, Laboratoire Paul Painlev\'e, F-59000 Lille, France}
\author{Jonathan  J. Thio}

\affiliation{Cavendish Lab., Department of Physics, Univ. of Cambridge, Cambridge CB3 0US, UK }
\author{Christopher K. Long}

\affiliation{Hitachi Cambridge Lab., J.J. Thomson Avenue, Cambridge CB3 0US, UK}
\affiliation{Cavendish Lab., Department of Physics, Univ. of Cambridge, Cambridge CB3 0US, UK }
\author{David R. M. Arvidsson-Shukur}

\affiliation{Hitachi Cambridge Lab., J.J. Thomson Avenue, Cambridge CB3 0US, UK}

\begin{abstract}

A central problem in quantum information is determining quantum-classical boundaries. 
 In the quasiprobability framework, a state is called classical if it is represented by a quasiprobability distribution that is positive, and thus a probability distribution. In recent years, the Kirkwood-Dirac (KD) distributions have gained much interest due to their numerous applications in modern quantum-information research. A particular advantage of the KD distributions is that they can be defined with respect to arbitrary observables.  Here,  we show that if two $d$-dimensional observables are picked at random, the set of classical (positive) states of the resulting KD distribution is a minimal polytope of dimension $2(d-1)$ with  $2d$ explicitly known vertices. This implies minimality of the sets of KD-real observables, of KD-positive measurement elements and of KD-positivity-preserving unitaries. We show how these results have implications on robust observations of nonclassical phenomena, on classical simulations of quantum circuits, and on foundations of quantum theory.
\end{abstract}

\maketitle

\textit{Introduction:---}Quasiprobability distributions are useful tools to study and characterize quantum experiments. They are objects similar to probability distributions but can assume negative or nonreal values---hallmarks of quantum phenomena. The Wigner function, which represents quantum states with distributions over the eigenvalues of conjugate observables, has played a pivotal role in the development of quantum optics~\cite{leonhardt, Serafini}.  More generally, the Wigner function, and other  quasiprobability distributions, are practical tools for investigating the classical-quantum boundary in various quantum-information-processing tasks \cite{Veitch_2012, Veitch_2014, Lostaglio18, YungHalp19,Raussendorf20,  Arvidsson-Shukur2020, Lostaglio20, lupu2022negative}.
Such advantages can be obtained only if the quasiprobability distributions characterizing the states, observables and processes involved are not all positive \cite{Spekkens08, Schmid_2024, schmid2024kirkwooddirac}. This  highlights the importance of identifying the (convex) set of positive states, meaning those states that have positive quasiprobability distributions. 

For example, the only pure Wigner-positive states are the pure Gaussian states~\cite{hudson1974wigner, gross2006}. However, a similar characterization of the set of mixed Wigner-positive states is currently unknown  and the subject of significant research ~\cite{gross2006, mandilara2009extending, vanherstraetencerf2021,chabaudetal2021}. In particular, there exist Wigner positive states that are not convex combinations (mixtures) of gaussian states \cite{Genoni, gross2006, vanherstraetencerf2021, hertz_decoherence_2023}.

Recently, the Kirkwood-Dirac (KD) distributions \cite{Kirkwood33,Dirac45, arvidssonshukur2024properties} have  come to the foreground as versatile mathematical tools to study discrete-variable quantum-information processes. 
Consider the sets of projectors $\Acal=\{|a_i\rangle\langle a_i| \}$ and  $\Bcal=\{|b_j\rangle\langle b_j| \}$ associated with two bases $\{ \cket{a_{i}} \}$ and $\{ \cket{b_j} \}$ in a $d$-dimensional Hilbert space $\Hcal$. Typically, these bases are eigenbases of observables $A$ and $B$ determined by the physical context of interest. Unlike the observables used to define the Wigner function,  $A$ and $B$ need not be conjugate.  The KD distribution associates to each quantum state $\rho$ a quasiprobability distribution $Q(\rho)$ with entries given by
\begin{equation}\label{eq:Qdef}
Q_{ij}(\rho)=\langle b_j|a_i\rangle \langle a_i|\rho|b_j\rangle \in \mathbb{D} ,
\end{equation}
 where $\mathbb{D}$ is the closed complex unit disk. In analogy with the definition of Wigner-positive states, if $Q_{ij}(\rho)\in [0,1]$, for all $i,j$, such that it is a proper probability distribution, we call $\rho$ KD-positive, or classical. If $Q(\rho)$ takes negative or nonreal entries, we call $\rho$ KD nonpositive, or nonclassical.

The KD distributions' negativity or nonreality underlie  quantum advantages in diverse fields~\cite{arvidssonshukur2024properties, lostaglio2023kirkwood}.  
Examples include metrology \cite{Arvidsson-Shukur2020, lupu2022negative, Jenne22}, weak values \cite{Aharonov88,Dressel14, Pang14, Pusey14, Dressel15, Pang15},  thermodynamics \cite{YungHalp17, Lostaglio18, Lostaglio20,Levy20,lostaglio2023kirkwood} and information scrambling \cite{YungHalp17,YungHalp18,Razieh19,YungHalp19,YungHalp23}. Moreover, negative or nonreal KD distributions can signal generalized contextuality (a strict notion of nonclassicality) \cite{Hofmann15, Lostaglio20, Ji2024quantitative, arvidssonshukur2024properties,thio2024} and violations of Leggett-Garg inequalities (temporal Bell inequalities) \cite{Jordan06,Ruskov06,williams2008weak,Emary_2014}.

In light of its diverse applications, a full characterization of KD-positive states is warranted. However, such a characterization is, as for the Wigner-positive states, nontrivial. The difficulty is exacerbated by the KD distribution's dependence on the two arbitrary bases  in its definition.  A partial characterization of the KD-positive pure states is obtained in Refs.~\cite{arvidsson-shukuretal2021, debievre2021, DeB23, Xu22, XU2024,yangetal2023a}, where KD positivity is linked to quantum uncertainty.  Moreover,  initial results on the KD positivity of mixed states  were obtained in specific instances~\cite{langrenez2023characterizing,yangetal2024,xu2024a,debievre2025kirkwooddiracrepresentationassociatedfourier,langrenezetal2024a}, 
 which highlight the intricacies of the structure of the set of KD-positive states: This set, which we shall denote by $\EcalKDC$, is not necessarily a polytope, and if it is, its extremal points may not be pure states,  making them hard to identify. (See Fig. \ref{fig:inclusions}.)

In this Article, we fully characterize the set of KD-positive states of almost all KD distributions. For almost all choices of the bases, we show that all KD-positive states are mixtures of the bases' states: They are of the simplest and minimal form.  (See Theorem~\ref{thm:principal}.)  Our result further implies that the set of observables, unitary transformations, and generalized measurements (defined below) with classical KD representations are also almost always minimal.  Moreover, the only unitaries that are KD-positivity preserving are (trivial) global phases. Finally, we discuss applications and implications of our results in relation to witness-construction for nonclassical phenomena,  classical simulations of quantum circuits, and foundations of quantum theory.

\begin{figure}
\begin{center}
\begin{tikzpicture}[scale=0.66]
	\draw (0,0) circle (2cm);
	\fill[pattern = crosshatch dots, pattern color=gray!50!white] (0,-2)--({sqrt(4-1.6*1.6)},-1.6)--({sqrt(4-1.55*1.55)},1.55)--(0,2)--({-sqrt(4-1.65*1.65)},1.65)--({-sqrt(4-1.65*1.65)},-1.65) --(0,-2);
	\draw[color=black, dashed](0,-2)--({sqrt(4-1.6*1.6)},-1.6)--({sqrt(4-1.55*1.55)},1.55)--(0,2)--({-sqrt(4-1.65*1.65)},1.65)--({-sqrt(4-1.65*1.65)},-1.65) --(0,-2);
	\node[scale=0.75] at (0,0) {$\convAB$};
    \fill (({-sqrt(4-1.65*1.65)},-1.65) circle (0.05); 
	\node[left,scale=0.75] at ({-sqrt(4-1.65*1.65)},-1.65){$\cket{a_1}\bra{a_1}$};
	\fill (0,-2) circle (0.05); 
	\node[below,scale=0.75] at  (0,-2){$\cket{a_2}\bra{a_2}$};
	\fill ({sqrt(4-1.6*1.6)},-1.6) circle (0.05); 
	\node[right,scale=0.75] at ({sqrt(4-1.6*1.6)},-1.6){$\cket{a_3}\bra{a_3}$};
	\fill ({-sqrt(4-1.65*1.65)},1.65) circle (0.05); 
	\node[left,scale=0.75] at ({-sqrt(4-1.65*1.65)},1.65){$\cket{b_1}\bra{b_1}$};
	\fill (0,2) circle (0.05); 
	\node[above,scale=0.75] at (0,2){$\cket{b_2}\bra{b_2}$};
	\fill ({sqrt(4-1.55*1.55)},1.55) circle (0.05); 
	\node[right,scale=0.75] at ({sqrt(4-1.55*1.55)},1.55){$\cket{b_3}\bra{b_3}$}; 
    \node[scale=0.75] at (-2,2.3){(a)}; 
\end{tikzpicture}
\begin{tikzpicture}[scale=0.66]
	\draw (0,0) circle (2cm);
	\fill[color=yellow!65!white] (0,-2)--({sqrt(4-1.6*1.6)},-1.6)-- (2,0) --({sqrt(4-1.55*1.55)},1.55)--(0,2)--({-sqrt(4-1.65*1.65)},1.65)--({-sqrt(4-1.65*1.65)},-1.65) --(0,-2);
    \fill[color=yellow!65!white] ({-sqrt(4-1.65*1.65)},1.65)-- (-1.43,1.08)--(-1.43,-1.08)--({-sqrt(4-1.65*1.65)},-1.65);
    \fill[pattern = north west lines, pattern color=black] ({-sqrt(4-1.65*1.65)},1.65)-- (-1.43,1.08)--(-1.43,-1.08)--({-sqrt(4-1.65*1.65)},-1.65);
    \filldraw[color=yellow!65!white] (-1.43,1.08) arc[start angle= 152.71, end angle=207.29,radius= 2.347 cm] ;
    \fill[ pattern = north west lines, pattern color=black] (-1.43,1.08) arc[start angle= 152.71, end angle=207.29,radius= 2.347 cm] ;
	\draw[color=black,dashed](0,-2)--({sqrt(4-1.6*1.6)},-1.6)-- (2,0) --({sqrt(4-1.55*1.55)},1.55)--(0,2)--({-sqrt(4-1.65*1.65)},1.65)-- (-1.43,1.08);
    \draw[color=black,dashed] (-1.43,-1.08)--({-sqrt(4-1.65*1.65)},-1.65) --(0,-2);
	\node[scale=0.75] at (0,0) {$\EcalKDC$};
    \draw[line width = 1] (-1.43,1.08) arc[start angle= 152.71, end angle=207.29,radius= 2.347 cm] ;
    \fill ({-sqrt(4-1.65*1.65)},-1.65) circle (0.05); 
	\node[left,scale=0.75] at ({-sqrt(4-1.65*1.65)},-1.65){$\cket{a_1}\bra{a_1}$};
	\fill (0,-2) circle (0.05); 
	\node[below,scale=0.75] at  (0,-2){$\cket{a_2}\bra{a_2}$};
	\fill ({sqrt(4-1.6*1.6)},-1.6) circle (0.05); 
	\node[right,scale=0.75] at ({sqrt(4-1.6*1.6)},-1.6){$\cket{a_3}\bra{a_3}$};
	\fill ({-sqrt(4-1.65*1.65)},1.65) circle (0.05); 
	\node[left,scale=0.75] at ({-sqrt(4-1.65*1.65)},1.65){$\cket{b_1}\bra{b_1}$};
	\fill (0,2) circle (0.05); 
	\node[above,scale=0.75] at (0,2){$\cket{b_2}\bra{b_2}$};
	\fill ({sqrt(4-1.55*1.55)},1.55) circle (0.05); 
	\node[right,scale=0.75] at ({sqrt(4-1.55*1.55)},1.55){$\cket{b_3}\bra{b_3}$}; 
	\node[left, scale=0.75] at  (-1.60,0){$C$};
	\fill (2,0) circle (0.05); 
	\node[right,scale=0.75] at  (2,0) {$D$};
    \node[scale=0.75] at (-2,2.3){(b)};
    \fill (-1.43,-1.08) circle (0.05);
    \fill (-1.43,1.08) circle (0.05); 
\end{tikzpicture}
    \caption{Schematic representation of  possible geometries of $\EcalKDC$, when $d=3$. The disc represents all states, with the boundary representing the pure states. (a) The simplest situation (Theorem \ref{thm:principal}): $\EcalKDC =\convAB$. (b) A more complex situation: $\convAB \subsetneq \conv{\EcalKDCpu } \subsetneq\EcalKDC$.  The point $D$ represents a pure state in  $\EcalKDC$. The arc $C$ represents a set of mixed extremal states in $\EcalKDC$. The hatched part represents the so-called \textit{exotic} states.}
    \label{fig:inclusions}
\end{center}
\end{figure}

\textit{The KD representation of quantum theory:---}
The KD distributions, as do other quasiprobability distributions, provide useful alternatives to the standard Schr\"odinger picture of quantum mechanics. If $\mab := \min |\braket{a_i}{b_j}| >0 $ (as we assume henceforth), the KD distribution is informationally complete~\cite{johansen2007}: Knowledge of $Q(\rho)$ uniquely determines $\rho$.  The marginals of  $Q(\rho)$ reproduce the Born rule:
\begin{eqnarray}
   & \nonumber \sum_j Q_{ij}(\rho)=\langle a_i|\rho|a_i\rangle, \quad \sum_i Q_{ij}(\rho)=\langle b_j|\rho|b_j\rangle \\
   &  \sum_{i,j} Q_{ij}(\rho)=1.
\end{eqnarray}
Therefore, $Q(\rho)$ behaves like a joint probability distribution for the observables $A$ and $B$, except that it may take complex values.  

To extend the KD formalism to a full representation of quantum mechanics, we introduce the KD symbol $W(F)$ of an operator $F$:
\begin{equation}
\label{eq:WtoQ}
    W_{ij}(F) = \frac{\matrixel{a_i}{F}{b_j}}{\innerproduct{a_i}{b_j}} = \frac{Q_{ij}(F)}{|\innerproduct{a_i}{b_j}|^2}.
\end{equation}
Here, we extended Eq.\eqref{eq:Qdef} by linearity to all operators $F$ on $\Hcal$.
Then, the expectation value $\Tr(F^{\dagger}\rho) $ of any operator $F$ with respect to any state $\rho$ is~\cite{langrenez2023characterizing,arvidssonshukur2024properties,schmid2024kirkwooddirac} 
\begin{equation}\label{eq:overlap}
\Tr(F^{\dagger}\rho)=\sum_{i,j} W_{ij}^*(F)Q_{ij}(\rho).
\end{equation}
Mathematically, each matrix element $W_{ij}(F)$ is a \textit{weak value} of $F$ \cite{Aharonov88}.
We call an operator KD real if its KD symbol is real. The real vector space of KD-real observables is denoted by $\VR$: $F\in\VR$ provided that $F=F^{\dagger}$ and that $W_{ij}(F)\in \R$ for all $i,j$. The set $\VR$ of KD-real observables is in general equally difficult to characterize as the set $\EcalKDC$ of KD-positive states.

A generalized measurement (also known as a positive-operator-valued measure) is described by a set $\{ M_k \}$, the elements of which are positive semi-definite  operators that satisfy   $\sum_k M_k = I_d $.
We denote by $M_{\textrm{KD}+}$ the space of measurement elements $M$ with  $W_{ij}(M)\geq 0$.

Together, the KD distribution $Q$ and the KD symbol $W$  provide a quasiprobability representation of quantum mechanics. Their restrictions to $\EcalKDC$, respectively $M_{\textrm{KD}+},$ provide a proper probability representation of a (classically describable) ``fragment’’ of quantum mechanics. The states of this fragment belong to $\EcalKDC$ and its observables are modeled with generalized measurements with elements in $M_{\textrm{KD} +}$.

\textit{Main result:---}{We now explore further the set of  KD-positive states and KD-real observables. Given two sets of basis projectors $\Acal$ and  $\Bcal$, one can readily confirm that
\begin{equation}\label{eq:inclusions}
\convAB\subseteq \EcalKDC\ \mathrm{and}\ \spanR (\Acal\cup \Bcal)\subseteq \VR.
\end{equation}
Here, $\conv{Z}$ denotes all convex combinations of the elements in $Z$, $\spanR(Z)$ denotes the real vector space generated by $Z$, and $\cup$ denotes the union of two sets. Our main result is:

\begin{Theorem}\label{thm:principal}
Let $\Acal=\{|a_i\rangle\langle a_i| \}$ and  $\Bcal=\{|b_j\rangle\langle b_j| \}$ be randomly chosen. Then, with probability $1$, the sets of KD-positive states, of KD-real observables, of KD-positive measurement elements and of KD-positivity-preserving unitaries $T$ (in $d >2$) are all of minimal size:
\begin{eqnarray}
\label{eq:goal}  \EcalKDC &=& \conv{\Acal\cup\Bcal}, \\ 
\label{eq:goal2} \VR &=& \spanRAB, \\ 
\label{eq:goal3} M_{\mathrm{KD} +} & \subset & \mathrm{span}_{\R^{+}}(\Acal\cup\Bcal), \\
\label{eq:goal4} T& = & \exp(i\alpha)\bbone, \alpha\in\R .
\end{eqnarray}
\end{Theorem}

Equation~\eqref{eq:goal} asserts that a state $\rho$ is KD positive if, and only if, there exist $0\leq \lambda_i, \mu_j\leq 1$ with $\sum_i\lambda_i+\sum_j\mu_j=1$, such that 
 \begin{equation}
     \rho=\sum_{i,j}\lambda_i |a_i\rangle\langle a_i| +\mu_j |b_j\rangle\langle b_j|.
 \end{equation}
Thus,  $\EcalKDC$ is 
a $2(d-1)$-dimensional polytope with  $\{ |a_i\rangle\langle a_i| \}$ and $\{ |b_j\rangle\langle b_j| \}$ as its $2d$  known vertices. $\EcalKDC$ cannot be ``smaller'' as any projector in $\Acal$ or in $\Bcal$ is KD positive by construction. 
Moreover, Eq.\eqref{eq:goal2} asserts that, with probability $1$, an observable is KD real if, and only if, there exist $\lambda_i,\mu_j\in\R$ such that 
  \begin{equation}
     F=\sum_{i,j}\lambda_i |a_i\rangle\langle a_i| +\mu_j |b_j\rangle\langle b_j|.
 \end{equation}
 Similarly, Eq.\eqref{eq:goal3} asserts that any measurement element $M \in M_{\textrm{KD}+}$ is of the form
  \begin{equation}
     M =\sum_{i,j}\lambda_i |a_i\rangle\langle a_i| +\mu_j |b_j\rangle\langle b_j|,
 \end{equation}
with $0 \leq \lambda_i, \mu_j \leq 1$. 
 Finally, Eq. \eqref{eq:goal4} asserts that 
 only trivial unitaries $T$  transform any KD-positive state $\rho$ into a KD-positive state $T\rho T^\dagger$.

 Before proving Theorem \ref{thm:principal}, we explain the precise meaning of the expression ``randomly chosen''. 
Whereas both sides of Eqs.~\eqref{eq:goal}--\eqref{eq:goal3} depend on the sets of basis projectors $\Acal$ and $\Bcal$, the relations themselves depend only on their transition matrix $U$ (with elements $U_{ij}=\langle a_i|b_j\rangle$), not on the individual basis elements. In other words, if $\Acal'$ and $ \Bcal'$ are such that $U'=U$, then, if Eqs.~\eqref{eq:goal}--\eqref{eq:goal3} hold for $\Acal$ and $\Bcal$, they  also hold for $\Acal'$ and $\Bcal'$~\cite{DeB23}. The unitary group $\Ucal(d)$ carries a unique invariant probability measure, called the Haar measure, that we denote by $\mu_{\textrm H}$~\cite{simon1996representations}. Theorem \ref{thm:principal} asserts that there is a probability-$1$ subset of transition matrices for which Eqs.~\eqref{eq:goal}--\eqref{eq:goal3} hold. 

We now turn to the proof of the theorem. We denote by $\Omega$ the set of all unitary matrices $U$ for which $\min 
 \left|U_{ij}\right|>0$; these are the unitary matrices that contain no zeroes.  The set $\Omega$ is open and $\mu_{\textrm H}(\Omega)=1$. In other words, $\Omega$ is a probability-$1$ subset of the unitary group. The latter statement follows directly from Proposition~\ref{prop:zero_sets_analytical} below.
 It was established in~\cite{langrenez2023characterizing} that, when $U\in\Omega$, $\dim(\spanRAB)=2d-1$ and
\begin{equation}
\EcalKDC=\conv{\Acal\cup\Bcal} \Leftrightarrow\VR=\spanRAB.
\end{equation}
In view of Eq.\eqref{eq:inclusions}, the theorem thus follows once we show that, with probability $1$, $\VR\subseteq \spanRAB$ or, equivalently, that $\dim\VR=2d-1$. 

\begin{figure}
    \centering
\begin{tikzpicture}[scale=0.53, transform shape]
    \node[black] at (0,3.4) {\Large Operators};
    \node at (0.1,-3.5) {\Large $\VR \stackrel{?}{=} \spanRAB$};
    
    \draw[black,thick] (0,0) circle (3);
    \node at (8,2.5) {\Large $\mcl{M}_{d}(\mathbb{C})$};
    \node at (0,2.5) {\Large $\mcl{S}(\H)$};

    \draw[prettyblue, dashed, thick] (0,0) ellipse (2.8 and 2);
    \node[prettyblue, dash pattern=on 3pt off 2pt on 1pt off 2pt] at (-1.3, 1) {\Large$\VR$};

    \draw[prettyred, densely dotted, thick] (0.7-0.4,-0.5) ellipse (2.2 and 1);
    \node[prettyred] at (0.7-0.4, -0.5) {\Large$\spanRAB$};

    \node[black] at (8,3.4) {\Large Weak-value matrices};
    \node at (7.7,-3.5) {\Large$W(\VR) = \ker C^U \stackrel{?}{=} \mathrm{span}_{\R}(\{X_k, Y_l\})$};
    
    \draw[black,thick] (8,0) circle (3);

    \draw[prettyblue, dashed, thick] (8,0) ellipse (2.8 and 2);
    \node[prettyblue] at (8, 1) {\Large$W(\VR) = \ker C^U$};

    \draw[prettyred, densely dotted, thick] (8.7-0.4,-0.5) ellipse (2.2 and 1);
    \node[prettyred] at (8.7-0.4, -0.5) {\Large$\mathrm{span}_{\R}(\{X_k, Y_l\})$};

    \draw[-stealth,black, thick] (3 + 0.2 ,0) -- (5 - 0.2,0);
    \node[black] at (4,0.3) {\Large$W$};

\end{tikzpicture}
    \caption{ The two circles represent the set of all operators and the corresponding set of KD symbols. Proving the Theorem reduces to showing that, with probability $1$, the red (dotted) and the blue (striped) sets coincide. The blue and red sets on the right have simpler mathematical descriptions than they do on the left, facilitating the proof of the Theorem.}
    \label{fig:schematic-overview}
\end{figure}
We start by noting that 
 \begin{equation}\label{eq:VKDr}
    \VR = \Ker(\Im W) \cap \SAH
\end{equation}
where $\cap$ denotes the intersection of two sets,  $\SAH$ denotes the real vector space of self-adjoint operators on $\mathcal H$, and $\Im{W} : F \in\mathcal{L}(\H) \to \Im{(W_{i,j}(F))}_{(i,j)\in\IntEnt{1}{d}^2}\in \C^{(d\times d)}$ denotes the imaginary components of the map $W$.
Indeed, an operator $F$ is KD real whenever Im$W(F)=0$, \emph{i.e.} $F\in \Ker(\Im{W})$.

As $Q$ is informationally complete (bijective), Eq.\eqref{eq:WtoQ} implies that the same holds for $W$.
That is, given an arbitrary matrix $w\in\MatC{d}$, there is a unique operator 
$F\in\mcl{L}(\H)$ such that 
\begin{equation}\label{eqn:op_from_wv}
    \matrixel{a_i}{F}{b_j} = w_{ij}\innerproduct{a_i}{b_j}.
\end{equation}
Hence
\begin{equation}
\dim \VR= \dim W(\VR).
\end{equation}
In what follows, we will show that $\dim W(\VR)=2d-1$ with probability $1$, thereby concluding the proof. We give a schematic overview of our proof strategy in Fig.~\ref{fig:schematic-overview}.

We proceed by showing how to characterize the real vector space $W(\VR)$ as a subspace of $\MatR{d}$, the space of $d \times d$ real matrices.  For that purpose, we first identify $W(\SAH)$. A direct computation, using the definition of $W$, shows that an operator $F$ is self-adjoint if, and only if, for all $i,j$,
\begin{equation}\label{eqn:self_adj_condition}
    \sum_l (W_{il}(F)-W_{jl}(F)^{*})U_{jl}^*U_{il}=0.
\end{equation}
It follows that $w\in W(\VR)$ if, and only if, $w\in\MatR{d}$ and, for all $i,j$, 
\begin{equation}
    \sum_l (w_{il}-w_{jl})U_{jl}^*U_{il}=0.
\end{equation}
This observation motivates us to define, for each unitary operator $U$,  the real linear map $C^U : \MatR{d}\to \SAO$, where
\begin{equation}\label{eqn:C_def}
    C^U_{jk}(w) = i\sum_l (w_{jl}-w_{kl})U_{kl}^*U_{jl}.
\end{equation}
One can therefore conclude that
\begin{equation}
W(\VR)=\Ker(C^U) .
\end{equation}
This implies that the two real vector spaces $\VR$ and $\Ker(C^U)$ are isomorphic and hence have the same dimension.  Additionally, in view of Eq.\eqref{eq:inclusions},
\begin{equation}
    W(\spanRAB)\subseteq W(\VR)=\Ker(C^U) .
\end{equation}
Hence, $\dim(\Ker(C^U))\geq 2d-1$. In other words, given $U$, the linear algebra question we need to address is: ``Is the dimension of $\Ker(C^U)$ equal to $2d-1$?''

We observe that 
\begin{equation}
    W(\spanRAB)\!=\!\spanR\{X_k, Y_\ell \},
\end{equation}
where $X_k=W(|a_k\rangle\langle a_k|)$ and $Y_k=W(|b_k\rangle\langle b_k|)$ are $d \times d$ matrices given by
\begin{eqnarray}
    X_k &=& \begin{pmatrix}
        0       &\cdots   &0\\
        \vdots  &\ddots  &\vdots\\
        0       &\cdots   &0\\
        1       &\cdots   &1\\
        0       &\dots   &0\\
        \vdots  &\ddots  &\vdots\\
        0       &\cdots   &0\\
    \end{pmatrix},
    \\
    Y_k &= &\begin{pmatrix}
        0       &\cdots  &0      &1      &0      &\cdots  &0\\
        \vdots  &\ddots &\vdots &\vdots &\vdots &\ddots &\vdots\\
        0       &\cdots  &0      &1      &0      &\cdots  &0\\
    \end{pmatrix}.
\end{eqnarray}
Here, $X_k$ is a matrix with ones on its $k$th row and zeroes elsewhere, and $Y_k=X_k^T$.
While both $\VR$ and the map $W$ depend on the two set of basis projectors $\Acal$ and $\Bcal$, the real vector space $W(\VR)=\Ker(C^U)$ depends only on their transition matrix $U$, not on the individual bases. In addition, $W(\spanRAB)$ depends neither on the sets of basis projectors $\Acal$ and $\Bcal$, nor even on the transition matrix $U$. This constitutes a considerable simplification that we will exploit to determine the dimension of $\Ker (C^U)$. 

For any $U$, we define
\begin{equation}
    \Triv:=\spanR\{X_k, Y_\ell \}\subseteq \Ker (C^U) .
\end{equation}
$\Triv$ is  well defined even if $U$ does have zeroes since $C^{U}$ in Eq. \eqref{eqn:C_def} is always well defined.
Let $\NT=\Triv^\perp$ denote the orthogonal complement of $\Triv$ in $\MatR{d}$ with respect to the Hilbert-Schmidt inner product on \MatC{d} given by $\inpr{A}{B}=\Tr (A^\dag B)$. Then, $\dim{\Ncal}=d^2 - (2d-1) = (d-1)^2$ and
\begin{equation}
    C^U:\MatR{d}=\Triv\oplus\NT\to \Ran{C^U}\subseteq\SAO .
\end{equation}
Here, $\Ran{Z}$ denotes the range of $Z$.
 Since $\Triv$ does not depend on $U$, the subspace  $\NT\subsetneq \MatR{d}$  is independent of $U$ as well.
We then define the set of ``good'' unitaries
\begin{equation}\label{eq:defGamma}
    \Gamma=\{U\in \Ucal(d)\mid \Ker(C^U)=\mathcal T\}.
\end{equation}
This terminology is justified by the fact that $\dim \Triv=2d-1$.
Thus, the proof of Eqs.~\eqref{eq:goal}--\eqref{eq:goal3} of Theorem~\ref{thm:principal} is reduced to the proof of the following proposition.
\begin{Prop}\label{prop:KerCU}
    The set $\Gamma$ of unitary matrices $U$ for which $\Ker(C^U)=\Triv$ is a set of probability $1$:
    \begin{equation}
    \mu_{\textrm H}(\Gamma)=1.
    \end{equation}
\end{Prop}
\noindent{\bf Proof of Proposition~\ref{prop:KerCU}.}
It follows from the rank-nullity theorem that $U\in\Gamma$ if, and only if,
\begin{equation}
C^U|_{\NT}:\NT\to \SAO
\end{equation}
is injective. We now show that this is indeed the case with probability $1$. 

Our proof relies on the existence of two sets of basis projectors $\Acal_{\star}$ and $\Bcal_{\star}$ for which $U_{\star} \in \Gamma$. As follows from a simple calculation in the Appendix~\ref{sec:App_Uex}, in  any dimension $d$, this property is satisfied by
\begin{equation}\label{eq:Ustar}
\Uex = \mathrm{exp}(id^{-1}J_{d}) = \id + zJ_{d}.
\end{equation}
Here,  $z= d^{-1}\left(e^{i}-1\right)$, $\id$ is the identity matrix, and $J_d$ is the matrix of ones.
We define 
\begin{equation}
\mathcal R_\star:=\Ran{C^{\Uex}}\subseteq\SAO.
\end{equation}
Since $\Uex\in\Gamma$, we know that $\dim(\mathcal R_\star)=\dim(\Ncal) = (d-1)^2$.
Let $\Pi_\star$ be the orthogonal projector onto $\mathcal R_\star$:
\begin{equation}
\Pi_\star:\SAO\to \mathcal R_\star\subsetneq \SAO, \quad \Pi_\star^2=\Pi_\star.
\end{equation} 
For every $U\in \Ucal(d)$ we then define
\begin{equation}
    \rcu := \Pi_\star  C^U|_{\NT} :\NT \to \mathcal R_\star .
\end{equation}
If $\rcu$ is a bijection, then $C^U|_{\NT}$ is an injection which in turn implies that $U\in\Gamma$. Moreover, by construction, $\rcu[\Uex]$ is a bijection. We now show that the unitaries $U\in\Ucal(d)$ for which $\tilde C^U$ is a bijection form a set  of probability $1$. This will conclude the proof. Consider two bases, one in $\Ncal$ and one in $\mathcal R_\star$. Because neither $\Ncal$, nor $\mathcal R_\star$ depend on $U$ these bases can be chosen independently of $U$.  We then denote the determinant of $\tilde C^U$ with respect to this choice of bases by $\det(\tilde C^U)$. Since the matrix elements of $C^U$ are quadratic polynomials in the real and imaginary parts of the matrix elements of $U$, it follows that the map
\begin{equation}
    c:\Ucal(d) \to \R, \; U\mapsto \det(\rcu)
\end{equation}
is also a polynomial in the real and imaginary parts of the entries of $U$. 
Since $c(U_\star) \neq 0$, the result now follows from Proposition~\ref{prop:zero_sets_analytical} below.
\qed

\begin{Prop}\label{prop:zero_sets_analytical}
    Let $f: \Ucal(d)\to \R$ be a polynomial in the real and imaginary matrix elements of $U$. Suppose that there exists at least one unitary matrix $\Uex$ for which $f(\Uex)\not=0$; then the zero-set of $f$ has vanishing Haar measure.
\end{Prop}
We provide a proof in Appendix~\ref{sec:Zeroes} (which contains references~\cite{hall2003,varadarajan1984,spivak1965,krantzparks2013,mityagin2020a,Rudin1987,Alfsen1963,abrahammarsden1978} therein). The proof relies on the fact that the zero sets of real analytic functions on $\R^m$ have vanishing Lebesgue measure. Finally, our proof yields as a corollary Eq.\eqref{eq:goal4}.

\textit{Discussion:---}We have shown that the convex set of KD-positive states, $\EcalKDC$, is, with probability $1$, simple and of minimal size. That is, $\EcalKDC$ is almost always the polytope of all mixtures of the states in $\Acal$ and $\Bcal$, of which the supporting hyperplanes are known \cite{langrenez2023characterizing}. 
Also the space of KD-real observables, the set of KD-positivity-preserving unitary transformations, and the set of KD-positive measurement elements are almost always minimal. Below, we discuss four applications of our results, with details in Appendices~\ref{sec:Qubits},~\ref{appendix:witness} and~\ref{appendix:simmulability}.

First, consider the task of preparing a state $\rho$ with the purpose of   generating nonclassical phenomena in any of the numerous fields where KD nonpositivity signals quantum advantages.  Naturally, one would want $\rho$ to be far from $\EcalKDC$, i.e., far from the shaded areas in Fig. \ref{fig:inclusions}.  
In general, it can be a formidable task to characterize  the region outside $\EcalKDC$. For example, the  KD distribution naturally associated to $n$ qubits ($d=2^n$) has $\mathcal{O}(d^{\log_{2}(d)})$ KD-positive pure extremal states and additional elusive mixed extremal states. Then $\EcalKDC$ is not necessarily a polytope and determining the distance of $\rho$ to this set is often unfeasible. However, our results show that almost all KD  distributions have only $2d$ explicitly known positive extremal states. For such KD distributions, one can lower bound the distance of any given $\rho$ to $\EcalKDC$ through the known witnesses associated with the facets of the simple polytope $\EcalKDC=\conv{\Acal\cup\Bcal}$.

Second, our theorem has applications to the foundational concept of generalized contextuality \cite{Spekkens2005}.  Generalizing Bell nonlocality, generalized contextuality describes the nonclassical phenomenon wherein experimentally indistinguishable procedures cannot be represented identically on a hidden-variable level. Prior works have established that experiments measuring the weak values in Eq. \eqref{eq:WtoQ} are noncontextual (classical) if, and only if, the input state $\rho$ is KD-positive \cite{Pusey14, kunjwal2019, thio2024}. Applying Theorem \ref{thm:principal}, we find that the set of input states that yield classical behavior is minimal and well-known for almost all observables $A$ and $B$.

Third, certain forms of nonclassical phenomena \cite{thio2024} stem from
mixed \textit{exotic} states: states that have a positive KD distribution but cannot be decomposed in terms of KD-positive pure states. [See Fig.~\ref{fig:inclusions}(b).] Our results show that such states can exist only in rare and  tailored KD distributions, such as those associated to certain spin-1 \cite{langrenez2023characterizing}  or $n$-qubit systems.

Fourth, Theorem~\ref{thm:principal}  has implications for the prospect of employing the KD distribution for efficient classical simulation of amenable quantum circuits. Quasiprobability distributions can be used to simulate quantum circuits using Monte Carlo sampling \cite{PhysRevLett.115.070501}. The runtime of such techniques scales polynomially with the accuracy but exponentially with the total KD nonpositivity of the circuit's input state, unitary gates, and measurements. 
 Our results immediately imply that, with probability 1, the only KD-positive circuit is an initialization in $\convAB$ followed immediately by a measurement in $M_{\textrm{KD}+}$. Consequently, only rare KD distributions, such as  the KD distribution in Appendix~\ref{sec:Qubits} \cite{x819-898d}, can be expected to be of use in standard protocols for efficient quantum-circuit simulation. 

To summarize, our results prove that the characterization of nonclassical features connected with KD distributions is almost always vastly simpler than it can be. Our results also show that rare KD distributions must be tailored to unveil certain aspects of quantum theory. Such KD tailoring will benefit from further studies of the map $C^U$.

\medskip

\textit{Acknowledgements:---} SDB and CL acknowledge the support of the CDP C2EMPI, as well as the French State under the France-2030 programme, the University of Lille, the Initiative of Excellence of the University of Lille, the European Metropolis of Lille for their funding and support of the R-CDP-24-004-C2EMPI project. SDB and CL also acknowledge the support of  the CNRS through the MITI interdisciplinary programs. WS was supported by the EPRSC and Hitachi. JJT was supported by the Cambridge Trust. CKL was supported by Hitachi. SDB thanks Girton College, where part of this work was performed, for its hospitality. The authors thank L. Flaminio, D. Radchenko and N. Yunger Halpern for helpful discussions.

\medskip

\textit{All authors contributed equally. The author ordering was randomized.}

\bibliography{references}

\appendix
\renewcommand{\theequation}{\thesection\arabic{equation}}
\renewcommand{\theHequation}{\thesection.\arabic{equation}}

\section{\texorpdfstring{On the zeroes of a non-zero polynomial function on $\Ucal(d)$}{On the zeroes of a non-zero polynomial function on U(d)}}\label{sec:Zeroes}
For convenience, we recall here Proposition~\ref{prop:zero_sets_analytical}. We expect this result to be well known, but we didn't find a proof in the literature. 
\addtocounter{Prop}{-1}
\begin{Prop}
    Let $f: \Ucal(d)\to \R$ be a polynomial in the real and imaginary matrix elements of $U$. Suppose that there exists at least one unitary matrix $\Uex$ for which $f(\Uex)\not=0$; then the zero-set of $f$ has vanishing Haar measure.
\end{Prop}

We first recall some known properties of the unitary group $\Ucal(d)$. 
We then show that the zero set of an analytic function on a connected analytic manifold has zero measure, 
implying Proposition~\ref{prop:zero_sets_analytical}. Note that throughout this Appendix, we call a function analytic if it is a real analytic function.

First, we need some basic facts on the unitary group $\Ucal(d)= \left\{U\in \MatC{d}, U^{\dag}U= \id\right\}$ that we collect here. We view a matrix $\Ucal\in U(d)$ as a tuple of $d$ columns:
\begin{equation}
    U=(U_1,\dots, U_d).
\end{equation}
The $U_j$ satisfy the orthonormality relations
\begin{equation}\label{eqn:orthomality}
U_j^{\dagger}U_k=\delta_{jk}.
\end{equation}
We write
\begin{equation}
U=R+iS,\quad \textrm{and}\quad  U_j=R_j+iS_j,
\end{equation}
for real and imaginary parts, and thus identify $U$ with an element of $\MatR{d}\times\MatR{d}\simeq\R^{2d^2}$. \\

The orthonormality relations of equation \eqref{eqn:orthomality} become, for all $j,k$:
\begin{equation}\label{eq:unitaryreal}
    R_j^TR_k+S_k^TS_j=\delta_{jk}, \quad R_j^TS_k-S_j^TR_k=0.
\end{equation}

\begin{Prop}\label{prop:Ud}
$\Ucal(d) $ is  compact and connected. Moreover, it is a $d^2$ dimensional analytic submanifold of $\MatR{d}\times\MatR{d}\simeq\R^{2d^2}$.
\end{Prop}
 The proof that $\Ucal(d)$ is compact and connected is standard, see for example~\cite{hall2003}. We did not find a proof of the analytic part of the proposition in the literature, although it is certainly well known. We provide an argument below. Let us first recall that a function $f:\R^n\to \R^m$ is analytic if it has, for each $x\in\R^n$, a convergent Taylor series expansion (in some neighbourhood of $x$). A submanifold of $\R^n$ is said to be a $k$-dimensonal analytic submanifold if it can be locally parametrized with analytic coordinates~\cite{varadarajan1984}.

Theorem~\ref{thm:Appendixsec5_1} states that a level surface of an analytic function $f$ is an analytic submanifold, provided $f$ has no critical points on the level surface. Theorem~\ref{thm:Appendixsec5_1} is the adaptation to the analytic setting of a well-known result in the $C^\infty$ setting, based on the analytic implicit function theorem; we omit the details (see Theorem 5.3 in~\cite{spivak1965} and Theorem 6.3.2 in~\cite{krantzparks2013}).

\begin{Theorem}\label{thm:Appendixsec5_1}
Let $(n,m)\in\N^2$ such that $m < n$. Let $f: \R^{n} \to \R^{m}$ be an analytic function and $k\in\R^m$. We denote by $\Sigma_{k} = \left\{ x\in \R^{n} \mid f(x)=k\right\}$. Suppose that for all $x\in \Sigma_{k}$, $Df(x)$ has rank $m$ (\textit{i.e}. is of maximal rank). Then, $\Sigma_k$ is an analytic submanifold of $\R^{n}$ with dimension $n-m$.
\end{Theorem}
We can now prove Proposition~\ref{prop:Ud}.\\
\noindent{\bf Proof of Proposition~\ref{prop:Ud}.}
    To prove that $\Ucal(d)$ is an analytic manifold of dimension $d^2$, we consider the map $f : 
             \MatR{d} \times \MatR{d}  \to \mcl{S}_{d}(\R) \times \mcl{A}_{d}(\R) $, with the notations used in Eq.\eqref{eq:unitaryreal}:
        \begin{equation}
        f(R,S) = \left(R^{T}R + S^{T}S, R^{T}S - S^{T}R\right).
    \end{equation}
Here, $\mathcal S_d(\R)$ is the set of real symmetric $d$ by $d$ matrices and $\mathcal A_d(\R)$ the set of real antisymmetric $d$ by $d$ matrices. Clearly, $f$ is analytic since it is a polynomial in the matrix elements of $R$ and $S$.
Consequently, using the notation in Theorem~\ref{thm:Appendixsec5_1}, one sees that $\Ucal(d)=\Sigma_{k}$ with $k=(\id,0_{d,d})$ where $0_{d,d}$ is the zero matrix of dimension $d \times d$ and $f$ is an analytic function.

Now, for a fixed unitary $U = R_{U} + i S_{U}$, we can define $g_{U} : \MatR{d} \times \MatR{d}  \to  \MatR{d} \times \MatR{d}$
\begin{equation}
g_{U}(R,S) = \left(R_{U}R - S_{U}S, R_{U}S + S_{U}R\right) .
\end{equation}
With this definition, one can check that $f\circ g_{U} = f$ for all $U\in \Ucal(d)$, meaning that $f$ is invariant under the left action of unitary matrices. This implies that it is enough to prove that the derivative of $f$ at  $(R,S)=(I_d,0_{d,d})$, $Df(\id, 0_{d,d})$, is of full rank. We find, for all $(H,K)\in \MatR{d}\times\MatR{d},\;$
\[
 Df(\id, 0_{d,d}) (H,K) = (H + H^{T}, K - K^{T}). 
\]

As any matrix $A$ in $\mcl{S}_{d}(\R)$ can be written as $\frac{1}{2}A + \frac{1}{2} A^{T}$ and any matrix $B\in \mcl{A}_{d}(\R)$ can be written as $\frac{1}{2}B - \frac{1}{2} B^{T}$, this proves that $Df(\id, 0_{d,d})$ is surjective in $\mcl{S}_{d}(\R) \times \mcl{A}_{d}(R)$, \textit{i.e.} $f$ is of full rank.

As $\MatR{d}\times\MatR{d} \simeq\R^{2d^2}$ and $\mcl{S}_{d}(\R) \times \mcl{A}_{d}(\R)\simeq\R^{d^2}$, we can apply Theorem~\ref{thm:Appendixsec5_1} to conclude that $\Ucal(d)$ is an analytic submanifold of $\MatR{d}\times\MatR{d}$ of dimension $d^2$.
\qed

We now turn to the results we need on the zeroes of analytic functions on compact, connected and analytic manifolds. 
\begin{Theorem}\label{thm:Appendixsec5_2}
    Let $\Omega$ be an open connected subset of $\R^{n}$ and $f: \Omega \to \C$ be an analytic function. Let $\mcl{Z}(f) = \left\{x\in \Omega \mid f(x)=0\right\} $ be the set of zeroes of $f$. Then either  $\mcl{Z}(f)$ has zero Lebesgue measure or $f$ vanishes identically on $\Omega$, meaning that $\mcl{Z}(f)=\Omega$. 
\end{Theorem}

A proof of this theorem can be found in~\cite{mityagin2020a}. To generalise this theorem, we define a set of zero measure on a manifold as follows. A subset $N$ of a manifold $M$ of dimension $m$ is said to have measure zero if on each local chart $\phi : V \to \Omega$, $\phi( V \cap N) $ is a subset of $\R^{m}$ with vanishing Lebesgue measure. This definition leads to the following theorem on connected and analytic manifolds.

\begin{Theorem}\label{thm:analytic_zero_meas}
    Let $M$ be a connected and analytic manifold of dimension $m$. Let $f: M \to \C$ be an analytic function and $\mcl{Z}(f) = \left\{x\in M \mid f(x)=0\right\} $ be the set of zeroes of f. Then, either $\mcl{Z}(f)$ has zero measure or $f$ vanishes identically on $M$, meaning that $\mcl{Z}(f) = M$.
\end{Theorem}
\begin{proof}
As $M$ is a manifold, it is second-countable and thus $M$ is covered by a countable set of charts $\{V_n\}_{n\in \N}$, i.e. $\cup_{{n\in \N}}V_n = M$. Let $\phi_n : V_{n} \to \Omega_n \subseteq \R^{m}$ denote the coordinate functions on the charts $V_n$.
Note that, if we fix $n\in \N$, $f\circ \phi_{n}^{-1} : \Omega_n \to \C$ is an analytic function on an open set of $\R^{m}$ (because $M$ is analytic). Consequently, by applying Theorem~\ref{thm:Appendixsec5_2}, either $\mcl{Z}(f \circ \phi_{n}^{-1}) $ has a zero Lebesgue measure or $\mcl{Z}(f\circ \phi_{n}^{-1})=\Omega_n$, meaning that $f$ vanishes on $V_{n}$.

We define two sets: $I$ is the set of $n\in \N$ for which $\mcl{Z}(f\circ \phi_{n}^{-1})$ has  zero Lebesgue measure and $J$ is the set of $n\in \N$ for which $f$ vanishes on $V_{n}$. Let $W = \cup_{i\in I} V_{i}$ and $Z = \cup_{j\in J} V_{j}$.
Note that $W,Z$ are the union of open sets, so are open. 
Moreover, as the sets $V_n$ cover $M$, we find $M = W \cup Z$. Additionally, suppose there exists some $p\in W\cap Z$, so that $p\in V_i\cap V_j$ for some $i\in I$ and $j\in J$. Note that, by definition of $J$, $f\circ \phi_{i}^{-1}$ vanishes on $\phi_i(V_i\cap V_j)\subseteq \Omega_i$. However, $\phi_i(V_j\cap V_i)$ is open in $\Omega_i$ and non-empty. Thus, it has positive measure in $\Omega_i$, so that $f\circ \phi_{i}^{-1}$ vanishes on a positive measure subset of $\Omega_i$. This contradicts the definition of $I$, and we deduce $W\cap Z = \emptyset$.

Therefore, as $M$ is connected, either $W = M$ or $Z=M$. If $Z = M$, then $f$ vanishes identically on $M$. If $M=W$, then for any $n\in \N$, $\phi_n(\mcl{Z}(f)\cap V_n) = \mcl{Z}(f\circ \phi_n^{-1})$ has zero Lebesgue measure. Then, for any chart $\phi:V\to\Omega\subseteq \R^m$,
\begin{equation*}
    \phi(\mcl{Z}(f)\cap V) = \bigcup_{\substack{n\in \N \\V_n \cap V \neq \emptyset}}\phi(\mcl{Z}(f)\cap V\cap V_n)
\end{equation*}
and thus,
\begin{equation*}
    \phi(\mcl{Z}(f)\cap V) 
    = \bigcup_{\substack{n\in \N \\V_n \cap V \neq \emptyset}} (\phi\circ \phi_n^{-1})[\mcl{Z}(f\circ \phi_n^{-1})\cap \phi_n(V\cap V_n)].
\end{equation*}
Moreover, for $n\in\N$, $(\phi\circ \phi_n^{-1})$ is analytic, so differentiable, and thus maps measure zero sets to measure zero sets (see \cite{Rudin1987} Lemma 7.25). Since the countable union of measure zero sets is itself measure zero, we deduce that $\phi(\mcl{Z}(f)\cap V)$ is measure zero. Since this is true for any chart $(V,\phi)$, we conclude that $\mcl{Z}(f)$ is measure zero in $M$.

\end{proof}

Applying Theorem \ref{thm:analytic_zero_meas} to the special case of $\Ucal(d)$, we obtain Proposition~\ref{prop:zero_sets_analytical}.
\begin{proof}
    Since $\Ucal(d)$ is a compact group, it has a unique measure that is left invariant, called the Haar measure \cite{Alfsen1963}. Moreover, $\Ucal(d)$ is a compact Lie group~\cite{varadarajan1984}. From Proposition 4.1.14 in~\cite{abrahammarsden1978}, the Haar measure on $\Ucal(d)$ is a volume form, meaning that the Haar measure is absolutely continuous with respect to the Lebesgue measure on each local coordinate patch. Thus, it suffices to show that $\mcl{Z}(f) := \{U\in \Ucal(d) : f(U) = 0\}$ has zero measure in the sense of Theorem~\ref{thm:analytic_zero_meas}. 
    
    Note that the matrix entries $U_{ij}, U_{ij}^*$ are analytic functions on $\Ucal(d)$ (by the implicit function theorem, Theorem 6.3.2 in \cite{krantzparks2013}). Since a polynomial of analytic functions is also analytic, we deduce that $f:\Ucal(d)\to\R$ is analytic. Finally, by assumption $f\neq 0$ and we apply Theorem~\ref{thm:analytic_zero_meas}.
\end{proof}

As a final remark, we explain how our results imply that the simple structure of $\EcalKDC$ is robust under small perturbations of the bases. Recall that $\Omega$ is defined as the set of all unitary matrices $U$ that contain no zeroes. Then, with $\Gamma$ as defined in Eq.\eqref{eq:defGamma}, the set $\Omega\cap\Gamma$ is open. Consequently, if $U$
 has no zeroes and is a ``good''  unitary ($U \in \Gamma$), then all sufficiently small perturbations of $U$ are also ``good''. Thus, the simple structure of $\EcalKDC$ is robust under small perturbations of the bases.  Conversely, since $\Omega\cap\Gamma$ has full measure, it is also dense and thus if two set of basis projectors $\Acal$ and $\Bcal$ are such that $\conv{\Acal\cup\Bcal}\subsetneq \EcalKDC$, then there exist nearby sets of basis projectors $\Acal'$ and $\Bcal'$ for which  $\conv{\Acal'\cup\Bcal'}= \EcalKDC$. In other words, there exists a small perturbation of the sets of basis projectors $\Acal$ and $\Bcal$ that do have a simple associated $\EcalKDC$.  

\section{\texorpdfstring{Proof that $\Uex$ satisfies $\mathrm{Ker}(C^{\Uex}) = \mathrm{span}_{\R}\left\{X_k,Y_{\ell}\mid k,\ell \in\IntEnt{1}{d}\right\}$.}{Proof that the map C corresponding to U star satisfies has the minimal kernel}}\label{sec:App_Uex}

 We recall that $\Uex$ is defined in Eq.\eqref{eq:Ustar}. First, as $\mcl{N}$ is defined as $\left(\mathrm{span}_{\R}\left\{X_k,Y_{\ell}\mid k,\ell \in\IntEnt{1}{d}\right\}\right)^{\perp}$, a matrix $N\in \NT$ satisfies the following equations:
\begin{eqnarray}\label{eq:NTsum}
\nonumber \forall i\in\IntEnt{1}{d}, \ 0 = \inpr{X_i}{N} = \sum_{j=1}^{d} N_{ij} \\ \forall j\in\IntEnt{1}{d}, \ 0 = \inpr{Y_j}{N} = \sum_{i=1}^{d} N_{ij},
\end{eqnarray}
where we recall that orthogonality is with respect to the Hilbert-Schmidt inner product $\langle A,B\rangle=\Tr(A^\dagger B)$. This means that a matrix $N$ is in $\NT$ if and only if the sum over  every column and over every row is zero.
Next, we consider
$\Uex = \mathrm{exp}(id^{-1}J_{d}) = \id + zJ_{d}$ where $z= d^{-1}\left(e^{i}-1\right)$, $\id$ is the identity matrix, and $J_d$ is the matrix of ones. Note that $\Re{(z)}\neq 0$ and $\Im{(z)}\neq 0$. We want to prove that $C^{\Uex}\left|_{\mcl{N}}\right.$ is injective. 
Suppose that there exists $N\in \NT$ with $C^{\Uex}(N)=0$. Then, for all $j,k$:
\begin{eqnarray}\label{eq:NTex_1}
\nonumber -iC^{\Uex}_{jk}(N)&=&\sum_{l=1}^{d} (N_{jl} - N_{kl})(\Uex)^{*}_{kl} (\Uex)_{jl} \\
\nonumber &=& \left|z\right|^{2} \sum_{l=1}^{d} (N_{jl} - N_{kl}) + (N_{jj} - N_{kj}){z^*} \\
&+&  (N_{jk} - N_{kk})z \\
& = & (N_{jj} - N_{kj}){z^*} +  (N_{jk} - N_{kk})z,
\end{eqnarray}
where the first term from the second line is zero as the sum over every column of $N$ is zero by Eq.\eqref{eq:NTsum}. Since $C^{\Uex}(N)=0$, the real and imaginary parts of Eq.\eqref{eq:NTex_1} are zero, yielding that, for all $(j,k)\in\IntEnt{1}{d}^2$:
\begin{eqnarray}\label{eq:NTex_2}
 (N_{jj} - N_{kj})\Re{(z)} &= & (N_{kk} - N_{jk})\Re{(z)} ; \\
\label{eq:NTex_3}  (N_{kj} - N_{jj})\Im{(z)} &= & (N_{kk} - N_{jk})\Im{(z)}.
\end{eqnarray}
As $\Re{(z)}$ and $\Im{(z)}$ are not zero, adding Eq.\eqref{eq:NTex_2} and \eqref{eq:NTex_3} yields that, for all $(j,k)\in\IntEnt{1}{d}^2$:
\begin{equation}\label{eq:NTex_4}
N_{kk} =N_{jk}. 
\end{equation}
Finally, for a fixed $k\in\IntEnt{1}{d}$,   Eq.\eqref{eq:NTsum} implies  that:
\begin{equation}
0 = \sum_{j=1}^{d} N_{jk} =  \sum_{j=1}^{d} N_{kk} = d  N_{kk}.
\end{equation}
Thus, for all $k\in\IntEnt{1}{d}$, $N_{kk}=0$.  Eq.\eqref{eq:NTex_4} then implies that $N=0$. Consequently, as $\Ker\left(C^{\Uex}\left|_{\NT}\right.\right) = \{0\}$, $C^{\Uex}\left|_{\NT}\right.$ is injective, concluding our proof. \qed

\section{On the absence of KD-positivity preserving unitaries.}

 Let $(\cket{a_i})_{i\in\IntEnt{1}{d}}$ and $(\cket{b_j})_{j\in\IntEnt{1}{d}}$ be two orthonormal bases of a Hilbert space of dimension $d$. We denote by $U$ the transition unitary matrix between the two bases and by $\EcalKDC$ the set of KD positive states associated to the unitary matrix $U$. We say that $T\in\mcl{U}(d)$ is a KD-positivity preserving unitary if and only if $T\EcalKDC T^{\dagger} = \EcalKDC$. 
The goal of this section is to prove the following proposition, regarding KD-preserving unitaries.

\addtocounter{Prop}{+1}
\begin{Prop}\label{prop:noKDpreserve}
    In any dimension $d\geqslant 3$, there exists a set $\mcl{V}(d)$ of full Haar measure among the set of unitaries $\mcl{U}(d)$  for which the only KD-positivity preserving unitaries $T$ are given by the family $\left(T(\theta) = e^{i\theta}I_d\right)_{\theta\in\R}$.
\end{Prop}
We now suppose that $U\in\mcl{W}(d)$, where $\mcl{W}(d) = \Omega \cap \Gamma$. Here $\Omega$ is the set of all unitaries that do not have zeroes, and $\Gamma$ is defined in Eq.\eqref{eq:Ustar}. We thus have that $\mab>0$ and $\EcalKDC=\convAB$. This implies that $\EcalKDCpu=\Acal \cup \Bcal$. As $T$ preserves the extreme points of $\EcalKDC$, $T$ realises a permutation of the set of $2d$ elements $\Acal\cup\Bcal$. The following lemma shows that only specific permutations are allowed.

\begin{Lemma}\label{lem:KDpres10}
    Suppose that $(\cket{a_i})_{i\in\IntEnt{1}{d}}$ and $(\cket{b_j})_{j\in\IntEnt{1}{d}}$ are two orthonormal bases such that $\mab > 0$ and $\EcalKDC = \convAB$. If $T$ is KD-positivity preserving, then one of the two alternatives occurs:
    \begin{enumerate}
        \item  There exist permutations $(\sigma, \tau)$ and phases $(\eta_{i})_{i\in\IntEnt{1}{d}},(\theta_{j})_{j\in\IntEnt{1}{d}}\in\R^{2d}$ such that
    \begin{eqnarray}
        \nonumber\forall i\in\IntEnt{1}{d},\quad T\cket{a_i} &=& e^{i\eta_{i}}\cket{a_{\sigma(i)}},\\
        \forall j\in\IntEnt{1}{d},\quad T\cket{b_j} &=& e^{i\theta_{j}}\cket{b_{\tau(j)}};
    \end{eqnarray} 
    \item  There exist permutations $(\sigma, \tau)$ and phases $(\eta_{i})_{i\in\IntEnt{1}{d}},(\theta_{j})_{j\in\IntEnt{1}{d}}\in\R^{2d}$ such that
    \begin{eqnarray}
        \nonumber \forall i\in\IntEnt{1}{d},\quad T\cket{a_i} &=& e^{i\eta_{i}}\cket{b_{\sigma(i)}}, \\
        \forall j\in\IntEnt{1}{d},\quad T\cket{b_j} &=& e^{i\theta_{j}}\cket{a_{\tau(j)}}.
    \end{eqnarray} 
    \end{enumerate}
\end{Lemma}
\begin{proof}
Suppose that we are not in case \textit{1}. By reordering the bases $(\cket{a_i})_{i\in\IntEnt{1}{d}}$ and $(\cket{b_j})_{j\in\IntEnt{1}{d}}$, we can suppose that $T\cket{a_1}=e^{i\eta_{1}}\cket{b_1}$, where $\eta_{1}\in\R$. Then, suppose that, for some $\ell\in\IntEnt{1}{d}$ and some $\eta_{2}\in\R$, $T\cket{a_2} = e^{i\eta_{2}}\cket{a_{l}}$. This implies that
\begin{equation}
    0=\bracket{a_1}{a_2} = \bra{a_2}T^{\dagger}T\cket{a_1} = e^{i(\eta_{1}-\eta_{2})}\bracket{a_{\ell}}{b_1}.
\end{equation}
Since $\mab > 0$, we obtain a contradiction. Moreover, as $T$ is unitary and preserves the set $\Acal\cup\Bcal$,  there exist $j\in\IntEnt{2}{d}$ and $\eta_{2}\in\R$ such that $T\cket{a_2} = e^{i\eta_{2}}\cket{b_j}$. Repeating this argument, we obtain that there exist a permutation $\sigma$ and phases $\left(\eta_{i}\right)_{i\in\IntEnt{1}{d}}$ such that
\begin{equation}
\forall i\in\IntEnt{1}{d},\ T\cket{a_i}=e^{i\eta_{i}}\cket{b_{\sigma(i)}}.
\end{equation}

Now, we turn our attention to $\left(T\cket{b_j}\right)_{j\in\IntEnt{1}{d}}$. Suppose that there exists $j\in\IntEnt{1}{d}$ and $\theta_{1}\in\R$ such that $T\cket{b_1} = e^{i\theta_{1}}\cket{b_j}$. Then, there exists $i\in\IntEnt{1}{d}$ such that $\sigma(i)\neq j$. Thus, 
\begin{equation}
    \bracket{a_i}{b_1} = \bra{a_i}T^{\dagger}T\cket{b_1} = e^{i(\theta_{1}-\eta_{i})} \bracket{b_{\sigma(i)}}{b_j}=0,
\end{equation}
which contradicts $\mab > 0$. Thus, $T\Bcal T^{\dagger} = \Acal$. As above, there exist a permutation $\tau$ and phases $(\theta_{j})_{j\in\IntEnt{1}{d}}\in\R^{d}$ such that
\begin{equation}
    \forall j\in\IntEnt{1}{d},\ T\cket{b_j}=e^{i\theta_{j}}\cket{a_{\tau(j)}}.
\end{equation}
This proves that we are in case \textit{2}.
\end{proof}

\begin{proof}[Proof of Proposition~\ref{prop:noKDpreserve}]
    Let $U\in\mcl{W}(d)$. Suppose that there exists a KD-positivity preserving unitary $T$. By Lemma~\ref{lem:KDpres10}, we have to consider two cases. In case \textit{1}, there exist two permutations $(\sigma,\tau)$ and $2d$ phases $(\eta_{i})_{i\in\IntEnt{1}{d}},(\theta_{j})_{j\in\IntEnt{1}{d}}$ such that
    \begin{eqnarray}
        \nonumber \forall i\in\IntEnt{1}{d},\quad T\cket{a_i} = e^{i\eta_{i}}\cket{a_{\sigma(i)}} \\
        \label{eq:defV1} \forall j\in\IntEnt{1}{d},\quad T\cket{b_j} = e^{i\theta_{j}}\cket{b_{\tau(j)}}.
    \end{eqnarray}
Using Eq.\eqref{eq:defV1}, we obtain that for all $(j,k)\in\IntEnt{1}{d}^2$:
    \begin{equation}\label{eq:matrix1}
         e^{i\eta_{k}}\bracket{a_k}{b_j} = e^{i\theta_{j}}\bracket{a_{\sigma(k)}}{b_{\tau(j)}}.
    \end{equation}
Suppose that $(\sigma,\tau)\neq (\mathrm{id},\mathrm{id})$. Without loss of generality, we can suppose that $\sigma \neq \mathrm{id}$. Thus, there exists $p\in\IntEnt{1}{d}$ such that $\sigma(p)\neq p$. Relations~\eqref{eq:matrix1} then imply that the set of unitary matrices $U\in \mcl{W}(d)$ for which the unitary $T$ is KD-positivity preserving is included in $\mcl{Z}\left(f_{\sigma,\tau}\right)$ where
\begin{equation}
    f_{\sigma,\tau} : U\in\mcl{U}(d) \to \left|U_{p1}\right|^2 - \left|U_{\sigma(p)\tau(1)}\right|^2\in\R.
\end{equation}
We recall that $\mcl{Z}\left(f_{\sigma,\tau}\right)$ denotes the set of zeroes of the function $f_{\sigma,\tau}$. Note that $f_{\sigma,\tau}$ is a polynomial in the real and imaginary parts of the elements of the matrix. 

We now construct a unitary matrix $U_{\star,1}$ such that $f_{\sigma,\tau}(U_{\star,1})\neq 0$. We define two orthonormal vectors of $\C^{d}$ given by:
\begin{equation}
    \widetilde{u}^{1}_{\star} = \frac{1}{\sqrt{d(d-1)}}\left(1 \cdots 1 -(d-1) 1 \cdots 1\right)^{\mathrm{T}}
    \end{equation}
and
\begin{equation}\widetilde{u}^{2}_{\star} = \frac{1}{\sqrt{d}}\left(1 \cdots 1\right)^{\mathrm{T}}
\end{equation}
where the coefficient $-\sqrt{\frac{d-1}{d}}$ is in place $p$. By the Gram-Schmidt process, there exists $(\widetilde{u}^{j}_{\star})_{j\in\IntEnt{3}{d}}$ such that $(\widetilde{u}^{j}_{\star})_{j\in\IntEnt{1}{d}}$ is an orthonormal basis of $\C^{d}$. We now define $U_{\star,1}$ to be 

\begin{equation}
     \begin{cases}
    \begin{psmallmatrix}
        \widetilde{u}^{1}_{\star} & \widetilde{u}^{\tau(1)}_{\star} &\widetilde{u}^{3}_{\star} & \cdots & \widetilde{u}^{\tau(1)-1}_{\star} & \widetilde{u}^{2}_{\star} &&\widetilde{u}^{\tau(1)+1}_{\star} & \cdots & \widetilde{u}^{d}_{\star}
    \end{psmallmatrix}&\mathrm{if \ }\tau(1)\ne 1\\
    \begin{psmallmatrix}
        \widetilde{u}^{1}_{\star} & \cdots & \widetilde{u}^{d}_{\star}
    \end{psmallmatrix}&\mathrm{if \ }\tau(1)=1
    \end{cases}.
\end{equation}

We then compute
\begin{equation}
    f_{\sigma,\tau}(U_{\star,1}) = \left\{\begin{array}{cc}
        \frac{d-2}{d} \neq 0& \mathrm{if \ } \tau(1)\neq 1 \\
         \frac{2-d}{d(d-1)} \neq 0& \mathrm{if \ } \tau(1)=1
    \end{array}\right. .
\end{equation}
Consequently, $\mcl{Z}\left(f_{\sigma,\tau}\right)$ is, by Proposition 2.4, a set of zero Haar measure.
To conclude case \textit{1}, suppose that $(\sigma,\tau)=(\mathrm{id},\mathrm{id})$. The relations~\eqref{eq:matrix1} imply that
\begin{equation}
    \forall (j,k)\in\IntEnt{1}{d}^2, e^{i\eta_{j}}=e^{i\theta_{k}}.
\end{equation}
We thus get that, according to the definition of $T$ given by equations~\eqref{eq:defV1}, $\theta_{j}=\eta_{1}~~[2\pi]$ and $\eta_{j} = \eta_{1}~[2\pi]$ for all $j\in\IntEnt{1}{d}$. Then, $T = e^{i\eta_{1}}I_d$.

In case \textit{2}, there exist two permutations $(\sigma,\tau)\in\Sigma(d)$ and $2d$ phases $(\eta_{i})_{i\in\IntEnt{1}{d}}$ and $(\theta_{j})_{j\in\IntEnt{1}{d}}$ such that
    \begin{eqnarray}
        \nonumber \forall i\in\IntEnt{1}{d}, & T\cket{a_i} = e^{i\eta_{i}}\cket{b_{\sigma(i)}}, \\
        \label{eq:defV2}\forall j\in\IntEnt{1}{d}, & T\cket{b_j} = e^{i\theta_{j}}\cket{a_{\tau(j)}}.
    \end{eqnarray}
We obtain similarly that, for all $(j,k)\in\IntEnt{1}{d}$:
\begin{equation}\label{eq:matrix2}
        e^{i\eta_{k}}\bracket{a_k}{b_j} = e^{i\theta_{j}}\bracket{b_{\sigma(k)}}{a_{\tau(j)}}.
    \end{equation}
 Note that there exists $p\in\IntEnt{1}{d}$ such that $\sigma(p)\neq 1$. Thus, relations~\eqref{eq:matrix2} imply that the set of unitary matrices $U\in \mcl{W}(d)$ for which the unitary $T$ is KD-positivity preserving is included in $\mcl{Z}\left(g_{\sigma,\tau}\right)$ where
\begin{equation}
    g_{\sigma,\tau} : U\in\mcl{U}(d) \to \left|U_{p1}\right|^2 - \left|U_{\tau(1)\sigma(p)}\right|^2\in\R.
\end{equation}
Note that $g_{\sigma,\tau}$ is a polynomial in the real and imaginary parts of the elements of the matrix. 

We now construct a unitary matrix $U_{\star,2}$ such that $g_{\sigma,\tau}(U_{\star,2})\neq 0$. We define $U_{\star,2}$ to be 
\begin{equation}
U_{\star,2} = \begin{psmallmatrix}
        \widetilde{u}^{1}_{\star} & \widetilde{u}^{\sigma(p)}_{\star} &\widetilde{u}^{3}_{\star} & \cdots & \widetilde{u}^{\sigma(p)-1}_{\star} & \widetilde{u}^{2}_{\star} & \widetilde{u}^{\sigma(p)+1}_{\star} & \cdots & \widetilde{u}^{d}_{\star}
    \end{psmallmatrix}.
\end{equation}
We then compute
\begin{equation}
    g_{\sigma,\tau}(U_{\star,2}) = \frac{d-2}{d} \neq 0.
\end{equation}
Consequently, $\mcl{Z}\left(g_{\sigma,\tau}\right)$ is, by Proposition 2.4, a set of zero Haar measure.

Consequently, the set of matrices $U$ that admit a KD-preserving unitary that is not a global phase change is included in the set
\begin{equation}
    \Lambda(d) = \mcl{W}(d)^{\mathrm{C}} \cup \bigcup_{\substack{(\sigma,\tau) \\ (\sigma,\tau)\neq (\mathrm{id},\mathrm{id})}} \mcl{Z}\left(f_{\sigma,\tau}\right) \  \cup \bigcup_{(\sigma,\tau)} \mcl{Z}\left(g_{\sigma,\tau}\right). 
\end{equation}
As $\mcl{W}(d)$ is of full Haar measure, we obtain that $\Lambda(d)$ is of zero Haar measure. Defining $\mcl{V}(d)$ to be the complementary set of $\Lambda(d)$, we conclude the proof.  
\end{proof}
This proposition thus proves that, with probability one in any dimension $d\geqslant 3$, the only KD-positivity preserving unitaries are given by $(e^{i\theta}I_d)_{\theta\in\R}$ and these unitaries act trivially on $\Acal\cup\Bcal$.

\section{\texorpdfstring{KD distribution on an $n$-qubit system}{KD distribution on an n-qubit system}}\label{sec:Qubits}
We have shown that, with probability one, the set of KD-positive states $\EcalKDC$ is an explicitly known polytope with $2d$ vertices. In this section we show through an example that in specific cases, the structure of $\EcalKDC$ can be much more involved.

We consider for that purpose a  KD distribution naturally associated to an $n$-qubit system. For $n=1$, we consider the computational basis $\mathbf{A}_{1} = \left(\cket{0},\cket{1}\right)$ and the basis $\mathbf{B}_{1} = \left(\cket{+},\cket{-}\right)$. The transition matrix between these two bases is the real Hadamard matrix, which is also the discrete Fourier Transform (DFT) matrix. The only KD-positive pure states are in that case the basis states. 

If $n\in\N^{*}$, we consider the KD representation of a system of $n$ qubits given by the following two tensor product bases :
\begin{eqnarray}
    \nonumber\mathbf{A}_{n} &=& \left\{\bigotimes_{i=1}^{n}\cket{\epsilon_i} \mid \epsilon_{i}\in\left\{0,1\right\}\right\} \\ \mathbf{B}_{n} &=& \left\{\bigotimes_{i=1}^{n}\cket{\eta_i} \mid \eta_{i}\in\left\{-,+\right\}\right\}.
\end{eqnarray}
This KD representation can be viewed as being associated to the group $\left(\Z_{2}\right)^{n}$ and its Heisenberg group as described in~\cite{debievre2025kirkwooddiracrepresentationassociatedfourier}. This allows one to identify the set of KD-positive pure states, which we denote by $\mcl{E}_{\mathrm{KD},n}^{\mathrm{pure}}$. The case $n=2$ is explicitly worked out below.  One can then determine the  cardinality $\mid\mcl{E}_{\mathrm{KD},n}^{\mathrm{pure}}\mid$ of $\mcl{E}_{\mathrm{KD},n}^{\mathrm{pure}}$:
\begin{equation}
\left|\mcl{E}_{\mathrm{KD},n}^{\mathrm{pure}} \right| = \sum_{p=0}^{n} 2^{n}\frac{\prod_{k=1}^{p} (2^{n} - 2^{k-1})}{\prod_{k=1}^{p} (2^{p} - 2^{k-1})}.
\end{equation}
One can prove the following inequalities  :
\begin{equation}
    2^{\frac{n^2}{4}}\leqslant \left|\mcl{E}_{\mathrm{KD},n}^{\mathrm{pure}} \right| \leqslant 2^{(n+1)^2+1}.
\end{equation}
This proves that the KD distribution naturally associated to $n$ qubits has $O(d^{\log_{2}(d)})$ KD-positive pure states, where $d=2^{n}$ is the dimension of the underlying Hilbert space. This observation illustrates the complexity of the structure of $\conv{\mcl{E}_{\mathrm{KD},n}^{\mathrm{pure}}}$ in comparison with the simplest situation that occurs with probability $1$ where $\left|\EcalKDCpu \right| = \left|\Acal \cup \Bcal \right| = 2d$. 

When $n=2$, we list in Table~\ref{table:KDpos2} the $20$ KD-positive pure states. In this case, the underlying Hilbert space is of dimension $4$. The number of KD-positive pure states in the simplest situation would be $8$, occurring with probability $1$. 

\begin{table}[H]
    \centering
    \begin{tabular}{|c|c|c|c|}
         \hline
         \ $\cket{00}$\ & \ $\cket{01}$ \ & \ $\cket{10}$ \ & \ $\cket{11}$\ \\
         \hline
         \  $\frac{\cket{00} + \cket{01}}{\sqrt{2}}$ \ & \ $\frac{\cket{00} - \cket{01}}{\sqrt{2}}$ \ & \ $\frac{\cket{10} + \cket{11}}{\sqrt{2}}$ \ & \ $\frac{\cket{10} - \cket{11}}{\sqrt{2}}$ \ \\
         \hline
         \ $\frac{\cket{00} + \cket{10}}{\sqrt{2}}$ \ & \ $\frac{\cket{00} - \cket{10}}{\sqrt{2}}$\ & \ $\frac{\cket{01} + \cket{11}}{\sqrt{2}}$\ &\ $\frac{\cket{01} - \cket{11}}{\sqrt{2}}$\ \\
         \hline
         \ $\frac{\cket{00} + \cket{11}}{\sqrt{2}}$\ & \ $\frac{\cket{00} - \cket{11}}{\sqrt{2}}$\ &\ $\frac{\cket{01} + \cket{10}}{\sqrt{2}}$\ &\ $\frac{\cket{01} - \cket{10}}{\sqrt{2}}$\ \\
         \hline
         \ $\cket{++}$ \ & \ $\cket{+-}$ \  & \ $\cket{-+}$ \ & \ $\cket{--}$\ \\
         \hline
    \end{tabular}
    \caption{List of KD-positive pure states for the KD representation associated with $2$ qubits.}
    \label{table:KDpos2}
\end{table}

We note that, as shown in~\cite{debievre2025kirkwooddiracrepresentationassociatedfourier}, for the $2$ qubit KD-representation, we have that :
\begin{equation}
\conv{\mcl{E}_{\mathrm{KD},2}^{\mathrm{pure}}} \subsetneq \EcalKDC.
\end{equation}
In other words, in that case, the set of KD-positive states may not be a polytope and contains exotic states. We expect this to be true in all dimensions.

\section{\texorpdfstring{Witnesses of $\convAB$}{Witnesses of conv(A union B)}}
\label{appendix:witness}

In this section, we show how the simple structure of $\convAB$ can be used to construct estimates of the distance of a state $\rho_{\star}$ to $\convAB$. In~\cite{langrenez2023characterizing}, the authors proved that the facets of $\convAB$ are given by $\conv{\Acal_{i}\cup\Bcal_{j}}$ for all $(i,j)\in\IntEnt{1}{d}^2$ where 
\begin{eqnarray*}
    \Acal_{i} &=& \left\{\cket{a_\ell}\bra{a_\ell} \mid \ell\in\IntEnt{1}{d}\backslash \{i\}\right\}\ \\
    \Bcal_{j} &=& \left\{\cket{b_k}\bra{b_k} \mid k\in\IntEnt{1}{d}\backslash \{j\}\right\}.
\end{eqnarray*}
For a fixed $(i,j)\in\IntEnt{1}{d}^2$, the bounding plane associated to the facet is given by
\begin{eqnarray*}
    &\cket{a_p}\bra{a_p} + & \\\nonumber &\spanR{\left\{\cket{a_\ell}\bra{a_\ell}-\cket{a_p}\bra{a_p}, \cket{b_k}\bra{b_k}-\cket{a_p}\bra{a_p}\right\}}_{\substack{\ell\neq i\\ k\neq j}} &
\end{eqnarray*}
where $p = \min \IntEnt{1}{d}\backslash \{i\}$. This set is an affine hyperplane inside the set $E = \spanRAB \cap \left\{F\in\SAO \mid \Tr(F) =1\right\}$ and a bounding plane of $\convAB$. Thus, there exists $\eta_{ij}\in E$ such that :
\begin{equation}
    \forall \rho\in\convAB, \Tr(\eta_{ij}\rho) \leqslant \Tr(\eta_{ij}\cket{a_p}\bra{a_p}).
\end{equation}
Such witnesses can be explicitly constructed thanks to the simple structure of $\convAB$. As an example, when the transition matrix is the DFT matrix in dimension $3$, the witnesses are given by 
\begin{equation}
    \forall (i,j)\in\IntEnt{1}{3}^2, \eta_{ij} = I_3 - \cket{a_i}\bra{a_i} - \cket{b_j}\bra{b_j}.
\end{equation} 

Suppose that such a witness $\eta_{ij}$ is given and that $\rho_{\star}$ is a quantum state satisfying that $\Tr(\eta_{ij}\rho_{\star}) > \Tr(\eta_{ij}\cket{a_p}\bra{a_p})$. Then, 
\begin{equation}\label{eq:witnesses}
    0 < \frac{\Tr\left(\eta_{ij}\left(\rho_{\star} - \cket{a_p}\bra{a_p}\right)\right)}{\norme{\eta_{ij}}{2}} \leqslant \mathrm{d}\left(\rho_{\star},\convAB\right)
\end{equation}
where $\mathrm{d}\left(\rho_{\star},A\right)$ denotes the distance of $\rho_{\star}$ to the set $A$. 

One can then use Theorem~\ref{thm:principal} from the paper in the following manner. If the transition matrix is in $\Gamma\cap\Omega$, we know $\EcalKDC=\convAB$, so that we obtain a lower bound on the distance of $\rho_{\star}$ to $\EcalKDC$, given by
\begin{equation}
 \mathrm{d}\left(\rho_{\star},\EcalKDC\right)\geqslant \frac{\Tr(\eta_{ij}\rho_{\star}) - \Tr(\eta_{ij}\cket{a_p}\bra{a_p})}{\norme{\eta_{ij}}{2}}>0.
\end{equation}
Obtaining such a bound whenever $\convAB\subsetneq \EcalKDC$ is considerably more difficult because in that case $\EcalKDC$ may not be a polytope, and even when it is, it may be very difficult to determine its bounding planes. 

\section{Simulability}
\label{appendix:simmulability}

A KD-positive quantum circuit is composed of a KD-positive input state, KD-positive unitary gates and KD-positive measurement elements. KD-positive states and KD-positive measurement elements were defined in the main text. KD-positive gates are defined as follows. Let $T$ be a unitary on $\Hcal$, then there exist $K_{ij:i'j'}\in\C$ so that
\begin{equation}
Q(T\rho T^\dagger)_{ij}=\sum_{i'j'} K_{ij:i'j'}Q_{i'j'}(\rho),
\end{equation}
with
\begin{equation*}
    \sum_{ij} K_{ij:i'j'}=1, \forall i',j'.
\end{equation*}
A unitary $T$ is said to be KD-positive if all $K_{ij:i'j'}\geqslant 0$. As explained in the Discussion section, such circuits can be efficiently simulated classically. 

To see how Theorem~\ref{thm:principal} impacts the problem of efficient simulability, we first note that, clearly, a KD-positive unitary is KD-positivity preserving. Hence Theorem~\ref{thm:principal}  implies that, with probability one, the only quantum circuits that are KD-positive are those that have as input states (convex mixtures of) the basis states, followed immediately by a KD-positive measurement. 

\end{document}